\documentclass[journal,letterpaper]{IEEEtran}

\usepackage[utf8]{inputenc} 
\usepackage[T1]{fontenc}
\usepackage{url}
\usepackage{ifthen}
\usepackage{cite}
\usepackage{xcolor}
\usepackage{amssymb}
\usepackage{authblk}
\usepackage{graphicx}
\usepackage{caption}

\usepackage{tabularx}

\usepackage{setspace}

\usepackage[top=2.58cm,bottom=2.75cm,left=1.6cm,right=1.6cm]{geometry}
\usepackage{setspace}

\usepackage[cmex10]{amsmath} 

\newtheorem{theorem}{Theorem}

\newtheorem{definition}{Definition}

\begin{document}

\title{Revisiting Shannon's Source Coding Theorem with Distributional Uncertainty under the Nonlinear Expectation Theory}

	

\setlength{\affilsep}{0.1em}

\author[*]{Wen-Xuan Lang\thanks{This work is supported in part by the National Key R\&D Program of China (Grant No. 2023YFA1009604), in part by the Beijing Municipal Natural Science Foundation (Grant No. L242013), in part by the Natural Science Foundation of China (Grant No. 62461160306), in part by the Fundamental Research Funds for the Central Universities (Grant No. 050-63253087), and in part by the Open Project of Key Laboratory of Mathematics and Information Networks (Beijing University of Posts and Telecommunications), Ministry of Education, China (Grant No. KF202602). \textit{(Corresponding author: Shaoshi Yang)}}}
\author[$\dagger$]{Shaoshi Yang}
\author[$\dagger$]{Jianhua Zhang}
\author[$\ddagger$]{Zhiming Ma}

\affil[*]{\small School of Mathematical Sciences, Nankai University}
\affil[$\dagger$]{School of Information and Communication Engineering, Beijing University of Posts and Telecommunications}
\affil[$\ddagger$]{University of Chinese Academy of Sciences, Academy of Mathematics and Systems Science, Chinese Academy of Science  \authorcr Emails: langwx@amss.ac.cn, \{shaoshi.yang, jhzhang\}@bupt.edu.cn, mazm@amt.ac.cn}

\maketitle


\begingroup
\setstretch{0.95}

\begin{abstract}
   In classical information theory, a source is modeled by a single, precisely known probability distribution. However, in the increasingly complex communication networks full of unanticipated, nonstationary, and heterogeneous random events, the assumption of precise and well-defined probability distributions to describe random variables appears somewhat idealized. Therefore, it is important to characterize the \textit{uncertainty of distributions} of source messages, subject to relaxing the assumption of deterministic probability models for analyzing information sources in information theory. Based on the \textit{nonlinear expectation theory}, a novel axiomatical system that extends classical probability theory, this paper investigates the information sources whose distributions themselves are uncertain, and refers to them as \textit{uncertain-distribution sources}. We generalize the fundamental concept information entropy to \textit{nonlinear information entropy}, which describes the measurement of the amount of information contained in a uncertain-distribution source. By using the strong law of large numbers under sublinear expectation, we establish a \textit{nonlinear source coding theorem}, which not only shows that the nonlinear information entropy is the upper bound for the infimum of achievable coding rate of uncertain-distribution sources under the maximum error probability criterion, but also determines a cluster point of the coding rate of uncertain-distribution sources under the minimum error probability criterion. Our findings reveal that the introduction of nonlinear expectation theory allows for a more comprehensive understanding of information sources.
\end{abstract}


\section{Introduction}\label{sec1}

\subsection{Beyond the ``True Distribution'' Assumption}

The mathematical foundation of Shannon's information theory is probability theory, where every source is modeled as a random variable governed by a single, well-defined probability measure \cite{gallager1968information,book}. This model has served as the backbone of communication engineering for decades, but it rests on a critical premise that there exists an objective, stable, and perfectly knowable ``true distribution'' underpinning the source.

In many emerging fields, this premise is increasingly questioned. For instance, in high-frequency financial data, the statistical properties of returns shift continuously and cannot be captured by a single distribution \cite{peng2019nonlinear}; in biological networks, noise sources may exhibit persistent distributional variation that resists conventional modeling. In such systems, the difficulty is not that we lack knowledge of the true distribution, but rather that imposing any single distribution would be an oversimplification. Random variables in these domains exhibit what we call intrinsic distributional uncertainty, that is, they are inherently not describable by one probability measure.

As mentioned in \cite{spiegelhalter2024does}, probability is not an objective property of the world but rather a construct based on subjective judgments and assumptions. In many cases that, for instance, have only limited observations or small sample size, there may not be a single, fixed probability distribution that accurately captures the underlying randomness. This reality challenges the traditional approach of relying on deterministic probability models, which typically assume the existence of an underlying ``true'' probability distribution. Therefore, novel approaches are required for effectively dealing with the distributional uncertainty inherent in information sources.

\textit{Nonlinear expectation theory} is a novel axiomatical system established by Peng \cite{peng2007g,peng2019nonlinear,peng2017theory}. In this theory, a random variable is defined on a sublinear expectation space $(\Omega,\mathcal{H},\mathbb{E})$, where $\mathbb{E}$ is a sublinear functional that acts as a ``worst-case expectation'' over a family of probability measures $\{P_\theta\}_{\theta\in\Theta}$, i.e., $\mathbb{E}[\cdot] = \sup_{\theta\in\Theta} E_{P_\theta}[\cdot]$. Crucially, no single $P_\theta$ is the ``true'' probability measure, the family itself is the complete description of the random variable. 

\begin{table*}[tbp]
	\centering
	\caption{Comparison of different source coding approaches under distribution uncertainty.}
	\label{tab:source_coding_comparison}
	\begin{tabularx}{\textwidth}{|X|X|X|X|X|}
		\hline
		\textbf{Method} & \textbf{Core Assumption} & \textbf{Distribution Type} & \textbf{Theoretical Support} & \textbf{Advantage} \\ \hline
		Our work based on nonlinear expectation theory & Source has intrinsic distributional uncertainty, no single true distribution exists & No distribution, only distributions family as intrinsic description & New compression rate based on nonlinear expectation theory & Robust performance, can achieve classical compression rate \\ \hline
		Classical Shannon analysis method (e.g., \cite{shannon1948mathematical,renyi1961measures,marichal2002entropy,rao2004cumulative}) & Source distribution fully known & Single true distribution & Optimal compression rate for the given distribution & Precise analysis based on probability theory \\ \hline
		Universal coding (e.g., \cite{ziv1972coding1,ziv1972coding2,ziv1977universal,jalali2010universal}) & Source distribution unknown, but a true distribution exists & Single true distribution, or known to belong to a distribution family & Convergence guarantee to the compression rate of the true distribution & Practical algorithm for learning the unknown distribution \\ \hline
		Probabilistic uncertain-distribution coding (e.g., \cite{mark2019the}) & Sender knows true distribution; receiver has only an approximate prior & Single true distribution, or known to belong to a distribution family & Bounds on extra code length based on prior distribution & Can handle imperfect priors at the receiver \\ \hline
	\end{tabularx}
\end{table*}

In this paper, stimulated by the idea of nonlinear expectation theory, we consider the situation where the information sources have distributional uncertainty, which we call \textit{uncertain-distribution sources}. Based on nonlinear expectation theory, the corresponding uncertain probability distributions of an uncertain-distribution source is potentially characterized by $\{p_{\theta}\}_{\theta \in \Theta}$. We deduce a \textit{nonlinear information entropy} as the form of $\underset{\theta \in \Theta }{\sup} \sum _{x}p_{\theta }(x)\log \frac{1}{p_{\theta }(x)}$ from a set of new hypotheses. We also redefine the metric for characterizing the performance of source encoder/decoder by sublinear expectation. Then, based on the strong laws of large numbers under sublinear expectations \cite{zhang2023the}, we establish a nonlinear source coding theorem, which shows that the nonlinear information entropy is the upper bound of the achievable coding rate of uncertain-distribution sources under the maximum error probability criterion, and $\underset{\theta \in \Theta}{\inf} \sum _{x}p_{\theta }(x)\log \frac{1}{V(x)}$ is a cluster point of the coding rate of uncertain-distribution sources under the minimum error probability criterion. These results drop the ``true distribution'' postulate entirely, thus extending information theory into the domain of nonlinear expectations.

\subsection{From Deterministic Probability to Nonlinear Expectation}

Many existing works address the challenge of ``unknown distributions'' (such as \cite{mark2019the,amos1998reliable}, etc.) and have been highly successful and well-established in many practical applications. These approaches typically rely on the assumption of a ``prior'' distribution or the use of an empirical distribution derived from observed data. In contrast, nonlinear expectation theory offers a different approach by not relying on any assumed distribution at all. This distinction is crucial because it allows us to handle situations where no distributional information is available, a scenario that traditional probabilistic methods are not designed to manage.

In \cite{wen2025hypothesis}, we apply the nonlinear expectation theory to compute the extreme error probabilities in a detection problem with uncertain-distribution noise. This method allows us to calculate the maximum and minimum error probabilities by considering over the entire family of distributions, providing values that reflect the combined influence of all distributions within that family, and offering a more comprehensive understanding of the uncertainty inherent in the problem. In \cite{lang2026theoritical}, we modeled channel noise under distributional uncertainty and, by examining a binary detection problem, highlighted the essential differences of the approach based on nonlinear expectation theory from conventional probabilistic analysis.

The introduction of nonlinear expectation theory in our work is not intended to replace established probabilistic methods, but rather to address scenarios where traditional methods struggle, particularly when there is a complete absence of distributional information. For instance, in highly complex or randomly dynamic systems, such as those encountered in financial markets, biological networks, communication networks, or quantum information processing, the behavior of random variables may be influenced by numerous factors that are difficult to model explicitly. The increasing prevalence of complex random systems and the need to model uncertainty more accurately have led to the growing importance of the nonlinear expectation theory. By providing a more general and flexible framework, this theory enables researchers to better understand and handle uncertainty in modern probabilistic modeling. To highlight the novelty of our approach, we compare classical source coding, universal coding, probabilistic uncertain-distribution coding, and our proposed method as shown in Table~\ref{tab:source_coding_comparison}.

\section{Basic Notions of Nonlinear Expectation Theory}\label{sec2}

Peng \cite{peng2007g} established the nonlinear expectation theory, which provides a more general framework for modeling complex stochastic systems. We briefly review the major concepts of nonlinear expectation theory as follows. For details of the definitions and theorems, researchers can refer to \cite{peng2019nonlinear,peng2017theory}.

Let $\mathcal{H}$ be a linear space consisting of real valued functions defined on the sample space $\Omega$. $\mathcal{H}$ satisfies $c\in \mathcal{H}$ for each constant $c \in \mathbb{R}$ and $|X|\in \mathcal{H}$ if $X \in \mathcal{H}$. If $X_1,\cdots, X_n\in \mathcal{H}$, then $\phi(X_1,\cdots, X_n)\in \mathcal{H}$ for each $\phi \in \mathbb{L}^{\infty}(\mathbb{R}^n)$, where $\mathbb{L}^{\infty}(\mathbb{R}^n)$ denotes the space of bounded Borel-measurable functions. The functions in $\mathcal{H}$ are called random variables.

\begin{definition}\label{def:sublinear expectation}
	A sublinear expectation $\mathbb{E}:\mathcal{H} \rightarrow \mathbb{R}$ is a functional defined on the space $\mathcal{H}$ satisfying the following properties:
	\begin{enumerate}
		\item{\textit{Monotonicity}: $\mathbb{E}[X]\geq \mathbb{E}[Y]$, if $X\geq Y$.}
		\item{\textit{Constant preserving}: $\mathbb{E}[c]=c,\forall c \in \mathbb{R}$.}
		\item{\textit{Sub-additivity}: $\mathbb{E}[X+Y]\leq \mathbb{E}[X]+\mathbb{E}[Y]$, $\forall X,Y \in \mathcal{H}$.}
		\item {\textit{Positive homogeneity}: $\mathbb{E}[\lambda X]=\lambda \mathbb{E}[X]$, for $\lambda>0$.}
	\end{enumerate}
	The triplet $(\Omega,\mathcal{H},\mathbb{E})$ is called a sublinear expectation space. If $\mathbb{E}$ satisfies only 1) and 2), then $\mathbb{E}$ is called a nonlinear expectation and $(\Omega,\mathcal{H},\mathbb{E})$ is called a nonlinear expectation space.
\end{definition}

Let $(\Omega,\mathcal{H},\mathbb{E})$ be a sublinear expectation space. Based on the representation theorem of sublinear expectation (Theorem 1.2.1 and 1.2.2 in \cite{peng2019nonlinear}), we can assume $\mathbb{E}$ is the upper expectation with respect to a family of probability measures $\mathcal{P}=\{P_{\theta}\}_{\theta \in \Theta}$, i.e. $\mathbb{E}[\cdot]=\underset{\theta\in \Theta}{\sup}E_{P_{\theta}}[\cdot]$. We call the family $\mathcal{P}$ the uncertain probability measures associated with the sublinear expectation $\mathbb{E}$. For a given $n$-dimensional random variable $X$ defined on a sublinear expectation space $(\Omega,\mathcal{H},\mathbb{E})$, the probability measure family $\mathcal{P}$ gives rise to a family of probability distributions $\{ F_X(\theta,A)=P_{\theta}(X\in A) , \ A\in \mathcal{B}(\mathbb{R}^n) \}_{\theta\in \Theta}$. For notational convenience, we also denote this family of probability distributions as $\{p_{\theta}(X)\}_{\theta \in \Theta}$. In such a case, we claim the corresponding uncertain probability distributions of $X$ are $\left\{p_{\theta }(X)\right\}_{\theta \in \Theta }$.

In order to analyze information sources under the framework of nonlinear expectation theory, it is necessary to introduce the strong laws of large numbers under sublinear expectations \cite{zhang2023the}, which is shown as the following Theorem \ref{thm:SLLN}. Let $V(A)=\underset{\theta \in \Theta}{\sup}P_{\theta}(A)$ ,$v(A)=\underset{\theta \in \Theta}{\inf}P_{\theta}(A)$,\footnote{If interested in the details of the pair $(V,v)$ and the strong laws of large numbers under sublinear expectations, one can refer to \cite{zhang2023the}.} and 
\begin{equation}
	C_V(X):=\int_{0}^{\infty}V(X\geq t)dt+\int_{-\infty}^{0}[V(X\geq t)-1]dt.
\end{equation}

\begin{theorem}\label{thm:SLLN}
	Let $\{X_i\}_{i=1}^{\infty}$ be a sequence of IID random variables defined on a sublinear expectation space $(\Omega,\mathcal{H},\mathbb{E})$ where $\mathbb{E}[\cdot]=\underset{\theta \in \Theta}{\sup} E_{P_{\theta}}[\cdot]$. Suppose that $C_V(|X_1|)<\infty$, $\{P_{\theta}\}_{\theta\in \Theta}$ is a countably-dimensional weakly compact family of probability measures on $(\Omega,\sigma(\mathcal{H}))$ in the sense that, for any bounded $Y_1,Y_2,\cdots \in \mathcal{H}$ and any sequence $\{P_n\}\subset \{P_{\theta}\}_{\theta\in \Theta}$, there is a subsequence $\{n_k\}$ and a probability measure $P\in \{P_{\theta}\}_{\theta\in \Theta}$ such that
	\begin{equation}
	\lim_{k\rightarrow\infty} P_{n_k}(\phi(Y_1,\cdots,Y_d)) =P(\phi(Y_1,\cdots,Y_d)).
	\end{equation}
	Denote by $C(\{x_n\})$ the set of cluster points\footnote{A cluster point of a sequence $\mathcal S$ is a point $x$, to which there is a subsequence of $\mathcal S$ that converges.} of the sequence $\{x_n\}$. Then, for any $b\in [-\mathbb{E}[-X_1],\mathbb{E}[X_1]]$ we have
	\begin{equation}
	V\left(b\in C\left(\left\{\frac{\sum_{i=1}^{n}X_i}{n}\right\}\right)\right) = 1.
	\end{equation}
\end{theorem}

\section{Characterization of Uncertain-Distribution Information Sources}\label{sec3}

Within the framework of nonlinear expectation theory, we assume that the random variables describing the messages sent by information sources are defined on the sublinear expectation space $(\Omega,\mathcal{H},\mathbb{E})$. In this case the information source is referred to as \textit{uncertain-distribution source}, as shown in Fig. \ref{fig_1}. The implicit meaning is that there are higher-level uncertainties inherent in the distributions of the random variables describing information sources, which do not obey a deterministic probability distribution. Therefore, it is necessary to use families of probability distributions to characterize these individual random variables.

Because of the ambiguities associated with the probability models themselves, it is impossible to study uncertain-distribution source using a deterministic probability framework. As a result, the fundamental concept in classical information theory, i.e., information entropy, is not directly applicable, and needs to be reformulated under the framework of the nonlinear expectation theory.

To elaborate a little further, for a random variable defined on a sublinear expectation space, we aim to consider the amount of information it contains from the standpoint of the circumstance with the highest uncertainty. Let $\hat{H}(X)$ denote the measurement of the amount of information contained in $X$, where $X$ is a discrete random variable on a sublinear expectation space $(\Omega,\mathcal{H},\mathbb{E})$ and the corresponding uncertain probability distributions of $X$ are characterized by $\left\{p_{\theta }(X)\right\}_{\theta \in \Theta }$. In order to reduce the mismatching error between the probabilistic model of a system and the practical system itself, it is crucial to carefully analyze a series of random variables that have unknown distributions. Especially, if the family of distributions of a random variable contains the uniform distribution, then the uniform distribution should be used for analyzing the amount of information embedded in the random variable, because uniform distribution represents the highest uncertainty case. Furthermore, it makes intuitive sense that the amount of information embedded in a random variable increases with the size of the probability distribution family the random variable may follow. In light of the above discussions, it is reasonable to make the following assumptions:

\textit{Assumptions:}

\begin{enumerate}
	\item{$\hat{H}(X)\leq \underset{\theta \in \Theta }{\sup } \sum _{x}p_{\theta }(x)\log \frac{1}{p_{\theta }(x)}$ and $\hat{H}(X)$ is continuous with each distribution in the set $ \left\{p_{\theta }(X)\right\}_{\theta \in \Theta }$. }
	\item{If the number of outcomes for $X$ is $N$ and the uniform distribution $p(i)=\frac{1}{N},i=1,2,\cdots,N$, is a member of the family $\{p_{\theta }(X)\}_{\theta \in \Theta }$, then $\hat{H}\left(X\right)=\log N$.}
	\item{Let $X$ and $Y$ be two random variables defined on the sublinear expectation spaces $(\Omega,\mathcal{H},\mathbb{E}_1)$ and $(\Omega,\mathcal{H},\mathbb{E}_2)$, respectively. If the distribution of $Y$ is stronger than that of $X$, then $\hat{H}(X)\leq \hat{H}(Y)$.}
	\item{If a random choice\footnote{Here the meaning of ``random choice" can be understood as an action that has multiple possible outcomes, as detailed in Section 6 of \cite{shannon1948mathematical}.} can be broken down into two successive random choices, the original $\hat{H}$ corresponding to the single random choice should be no greater than the supremum of the weighted sum of individual values of $\hat{H}$  corresponding to the two successive random choices.}
\end{enumerate}

\begin{figure}[tbp]
	\centering
	\includegraphics[scale=0.36]{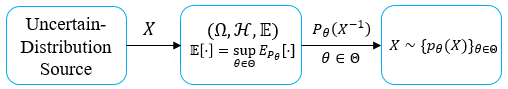}
	\captionsetup{font={scriptsize}}
	\caption{An illustrative example of uncertain-distribution source.}
	\label{fig_1}
\end{figure}

Based on the above four assumptions, one can prove the following theorem, which presents the mathematical expression of the amount of information contained in random variables defined on a sublinear expectation space. Without loss of generality, in this paper we consider only discrete random variables.

\begin{theorem}\label{theentropy}
	Let $X$ be a discrete random variable on a sublinear expectation space $(\Omega,\mathcal{H},\mathbb{E})$. The corresponding uncertain probability distributions of $X$ are $\left\{p_{\theta }(X)\right\}_{\theta \in \Theta }$. Then the amount of information contained in $X$ can be expressed as $\hat{H}\left(X\right)=\underset{\theta \in \Theta }{\sup } \sum _{x}p_{\theta }(x)\log \frac{1}{p_{\theta }(x)}$.
\end{theorem}

\textit{Proof sketch}: The proof follows Shannon's original counting argument but replaces fixed distributions by worst-case choices. When the distribution family is a singleton, it recovers the classical Shannon entropy. Details are omitted due to space and can be found in the extended version~\cite{lang2026nonlinear}.

\section{Nonlinear Source Coding Theorem}\label{sec4}

Similar to the case in classical information theory, for uncertain-distribution sources, encoding and decoding are necessary in order to reduce information redundancy and improve communication efficiency. In this context, \textit{nonlinear source coding} refers to the process of encoding and decoding for uncertain-distribution sources. It is important to note that uncertain-distribution sources are defined on a sublinear expectation space $(\Omega,\mathcal{H},\mathbb{E})$, which implies that the sources have distribution uncertainty. Therefore, the probability models for analyzing the associated problems are uncertain as well, and the performance measure of the transmission process need to be redefined.

The encoding and decoding in nonlinear source coding processes can still be represented as functions. Specifically, a $(||\mathcal{W}||,n,f_n,g_n)$ nonlinear source code consists of:
\begin{itemize}
	\item[$\bullet$] A source encoding function $f_n\colon \mathcal{X}^{n}\rightarrow \mathcal{W}$ that assigns an index $W \in \mathcal{W}$ to each $\boldsymbol{X}^n\in \mathcal{X}^n$.
	\item[$\bullet$] A source decoding function $g_n\colon \mathcal{W}\rightarrow \mathcal{X}^{n}$ that assigns an estimate $\hat{\boldsymbol{X}}^n\in \mathcal{X}^n$ to each index $W\in \mathcal{W}$.
\end{itemize}
The source coding rate $R_{\textrm{s}}$ is defined as $\frac{\log \|\mathcal{W} \|}{n}$.

Within the framework of nonlinear expectation theory, due to the existence of the family of uncertain probability measures $\{P_{\theta}\}_{\theta \in \Theta}$  associated with the sublinear expectation, the performance of the above nonlinear source code is measured by the maximum and the minimum probabilities of the event that the estimate of the message is different from the message actually sent. Here, these probabilities are calculated as
\begin{equation}
\mathbb{E}[I_{\{\hat{\boldsymbol{X}}^n\neq \boldsymbol{X}^n\}}]=\sup_{\theta \in \Theta } P_{\theta}(\hat{\boldsymbol{X}}^n\neq \boldsymbol{X}^n)=\sup_{\theta \in \Theta } P_{e,\theta}^{(n)},
\end{equation}
\begin{equation}
\mathcal{E}[I_{\{\hat{\boldsymbol{X}}^n\neq \boldsymbol{X}^n\}}]=\inf_{\theta \in \Theta } P_{\theta}(\hat{\boldsymbol{X}}^n\neq \boldsymbol{X}^n)=\inf_{\theta \in \Theta } P_{e,\theta}^{(n)}.
\end{equation}
Intuitively, the above mentioned terms ``maximum probability'' and ``minimum probability'' correspond to the \textit{conservative} and the \textit{aggressive} strategies of designing nonlinear source code, respectively. Based on this insight, there are two criteria for evaluating the performance of nonlinear source code, namely the maximum error probability criterion and the minimum error probability criterion. The former refers to ensuring $\mathbb{E}[I_{\{\hat{\boldsymbol{X}}^n\neq \boldsymbol{X}^n\}}]$ to be sufficiently small, while the latter refers to ensuring $\mathcal{E}[I_{\{\hat{\boldsymbol{X}}^n\neq \boldsymbol{X}^n\}}]$ to be sufficiently small.

For a discrete IID source sequence $X_{1},X_{2}, \cdots,X_{n},\cdots$ on a sublinear expectation space, it is important to study the fundamental limit of its source coding rate. Note that $\hat{H}(X_1) = \hat{H}(X_2) = \cdots = \hat{H}(X_n) = \cdots$ holds true, since we consider an IID source sequence. We also consider the case where $V$ satisfies the following uniform boundedness condition: there exists a constant $M<\infty$ such that for every $n\geq1$, $\sum_{\boldsymbol{x}^{n}\in \mathcal{X}^{n}}V(\boldsymbol{x}^{n})\leq M$.

\begin{theorem}\label{thesource}
	Let $X_{1},X_{2},\cdots,X_{n},\cdots$ be a discrete IID source sequence defined on a sublinear expectation space $(\Omega,\mathcal{H},\mathbb{E})$ where $\mathbb{E}[\cdot]=\underset{\theta \in \Theta}{\sup} E_{P_{\theta}}[\cdot]$. The corresponding uncertain probability distributions of $X_1$ are $\left\{p_{\theta }(X_1)\right\}_{\theta \in \Theta }$. $(V,v)$ is the pair of capacities generated by $\mathbb{E}$. Suppose that $C_V(|\log \frac{1}{V(X_i)}|)<\infty$, $\{P_{\theta}\}_{\theta\in \Theta}$ is a countably-dimensional weakly compact family of probability measures on $(\Omega,\sigma(\mathcal{H}))$ (see Theorem \ref{thm:SLLN}), and $V$ satisfies the uniform boundedness condition. Then:
		\begin{enumerate}
			\item {For any $R_{\textrm{s}} > \underset{\theta \in \Theta}{\inf} \underset{x}{\sum}p_{\theta }(x)\log \frac{1}{V(x)}$ and any $\epsilon$ of positive value, there is a sufficiently large $n$ and a $(\|\mathcal{W} \|,n,f_n,g_n)$ nonlinear source code with source coding rate $R_{\textrm{s}}$ such that $\mathcal{E}\left[I_{\{\hat{\boldsymbol{X}}^n\neq \boldsymbol{X}^n\}}\right]< \epsilon$;}
			\item {For any $R_{\textrm{s}} > \hat{H}(X_{1})$ and any $\epsilon$ of positive value, there is a sufficiently large $n$ and a $(\|\mathcal{W} \|,n,f_n,g_n)$ nonlinear source code with source coding rate $R_{\textrm{s}}$ such that $\mathbb{E}\left[I_{\{\hat{\boldsymbol{X}}^n\neq \boldsymbol{X}^n\}}\right]< \epsilon$.}
	\end{enumerate}
\end{theorem}

\textit{Proof}: Due to $X_{1},X_{2},\cdots,X_{n},\cdots$ is a discrete IID source, for any $i=2,...,n$, the uncertain probability distributions of $X_i$ are also denoted as $\{p_{\theta}(X)\}_{\theta \in \Theta}$.

For 1), let $Y_i:=\log \frac{1}{V(X_i)} \in \mathcal{H}$, and $S_n:= \sum_{i=1}^{n}Y_i $. For any $\mu \in [-\mathbb{E}[-Y_i],\mathbb{E}[Y_i]]$ and any $\epsilon>0$, based on the strong laws of large numbers under sublinear expectations, there exists a monotonically increasing sequence $\{n_k\}$ such that
\begin{equation}
V\left(\left\{\omega \Big| |\frac{S_{n_k}(\omega)}{n_k}-\mu| \leq \epsilon \right\}\right) \geq 1-\epsilon.
\end{equation}
Let $\mathcal{M}_{\epsilon}^{(k)}:=\{\boldsymbol{x}^{n_k} \big| |\frac{-\log V(\boldsymbol{x}^{n_k})}{n_k}-\mu| \leq \epsilon \}$. Then, there is $V\left(\mathcal{M}_{\epsilon}^{(k)}\right) \geq 1-\epsilon$, and for any $\boldsymbol{x}^{n_k} \in \mathcal{M}_{\epsilon}^{(k)}$,
\begin{equation}
2^{-n_k (\mu+\epsilon)} \leq V(\boldsymbol{x}^{n_k}) \leq  2^{-n_k (\mu-\epsilon)}.
\end{equation}

Therefore, we can get
\begin{equation}
(1-\epsilon)2^{n_k(\mu-\epsilon)} \leq  \|\mathcal{M}_{\epsilon}^{(k)} \| \leq M 2^{n_k(\mu+\epsilon)}.
\end{equation}

As a result, for any $\delta >0$, there is a sufficiently large $k$ and the source encoding function 
assigns a unique index to each sequence in $\mathcal{M}_{\epsilon}^{(k)}$. This source encoding function is an one-to-one mapping, whose outputs are easily decodable and satisfy $\mathcal{E}\left[I_{\{\hat{\boldsymbol{X}}^{n_k}\neq \boldsymbol{X}^{n_k}\}}\right] = v\left((\mathcal{M}_{\epsilon}^{(k)})^{c}\right) < \epsilon$, and the source coding rate satisfies $R_{\textrm{s}} \geq  \underset{\theta \in \Theta}{\inf} \sum _{x}p_{\theta }(x)\log \frac{1}{V(x)} +\epsilon'$. Based on the arbitrary of $\mu$ and $\epsilon$, we can get the conclusion.

For 2), let let $\gamma>0$ s.t. $\hat{H}(X_1)+2\gamma < R_s$, and $B_n$ be the set
\begin{equation}
\left\{\boldsymbol{x}^{n} \Big| H_{\hat{p}_{\boldsymbol{x}^{n}}}\leq \hat{H}(X_1)+\gamma \right\},
\end{equation}
where $\hat{p}_{\boldsymbol{x}^{n}}$ is the empirical distribution of $\boldsymbol{x}^{n}$. Then, there is $|B_n|\leq 2^{nR_s}$ for sufficiently large $n$. Based on the Sanov Theorem under sublinear expectation (Section 3.1 in \cite{fuqing2012relative}), we can get
\begin{equation}
\limsup_{n\rightarrow \infty} \frac{1}{n}\log V((B_n)^c)\leq -\inf_{q\in F}H_{V}(q)<0,
\end{equation}
where $F$ is the set $\{q| \underset{\theta\in \Theta}{\inf} \| q-p_{\theta} \| \geq \delta \}$.

As a result, for any $\epsilon >0$, there is a sufficiently large $n$ and the source encoding function 
assigns a unique index to each sequence in $B_n$. This nonlinear source code satisfies
\begin{equation}
\mathbb{E}\left[I_{\{\hat{\boldsymbol{X}}^{n}\neq \boldsymbol{X}^{n}\}}\right] = V\left((B_n)^{c}\right) < \epsilon.
\end{equation}

This concludes the proof of Theorem \ref{thesource}. $\hfill\blacksquare$

It is worth noting that the source coding rate limit in the first statement of Theorem \ref{thesource} and its classical counterpart yield different performance limits under the minimum error probability criterion. Note that this result does not contradict the classical source coding theorem. In classical information theory, if the random variable is assumed to follow a specific probability distribution from a family of distributions $\{p_{\theta}(X)\}_{\theta\in \Theta}$ but the exact one is unknown, then the limit of source coding rate can be optimized to $\underset{\theta \in \Theta}{\inf} \underset{x}{\sum}p_{\theta }(x)\log \frac{1}{p_{\theta }(x)}$ under the minimum error probability criterion. By contrast, in this paper, the more general scenario and the relaxation of the traditional i.i.d. assumption make it possible to exceed the traditional source coding rate limit. Because $V(x)$ is greater than or equal to $p_{\theta}(x)$ for any $\theta\in \Theta$, we know that the value of $\underset{\theta \in \Theta}{\inf} \underset{x}{\sum} p_{\theta }(x)\log \frac{1}{V(x)}$ in the first statement of Theorem \ref{thesource} is smaller than or equal to the value of $\underset{\theta \in \Theta}{\inf} \underset{x}{\sum} p_{\theta }(x)\log \frac{1}{p_{\theta }(x)}$.

\section{Examples of the Bernoulli Type Uncertain-Distribution Sources}\label{sec5}

Consider a uncertain-distribution source $X$ with two input values $\{0,1\}$. Suppose that the probability of the event $\{X=0\}$ is uncertain and takes value in the interval $\Theta=[\max(p-\epsilon , 0),\min(p+\epsilon,1)]$ and $\textrm{Pr}(X=1)=1-\textrm{Pr}(X=0)$. Then the uncertain probability distributions of $X$ are represented by a family of probability distributions $\{p_{q}(X)\}_{q\in \Theta}$, which is expressed as
\begin{equation}
\{p_{q}(X)\}_{q\in \Theta}=\{p_q=\{q,1-q\}|q\in \Theta \}.
\end{equation}

For the uncertain-distribution source $X$, we visualize the results of Theorem \ref{thesource} in Fig. \ref{fig_2}. The dash-dotted lines represent the cluster point of the coding rate of the uncertain-distribution source with $\epsilon=0.02$ and $\epsilon=0.03$ under the minimum error probability criterion, and are calculated by $\underset{\theta \in \Theta}{\inf} \underset{x}{\sum}p_{\theta }(x)\log \frac{1}{V(x)}$. The dashed lines represent the nonlinear information entropy of the uncertain-distribution sources with $\epsilon=0.02$ and $\epsilon=0.03$, and are calculated by $\underset{\theta \in \Theta }{\sup } \underset{x}{\sum}p_{\theta }(x)\log \frac{1}{p_{\theta}(x)}$. For the convenience of comparison, we also draw the classical Shannon entropy of the information source without distribution uncertainty, which is corresponding to the case $\epsilon=0$ and represented by the solid line in Fig. \ref{fig_2}. This numerical example confirms that the cluster point we identify can lie below the classical Shannon entropy and that, for sufficiently large code lengths, nonlinear source coding outperforms source coding without distribution uncertainty under the minimum error probability criterion, demonstrating the practical relevance of the nonlinear source coding theorem.
\begin{figure}[tbp]
	\centering
	\includegraphics[scale=0.5]{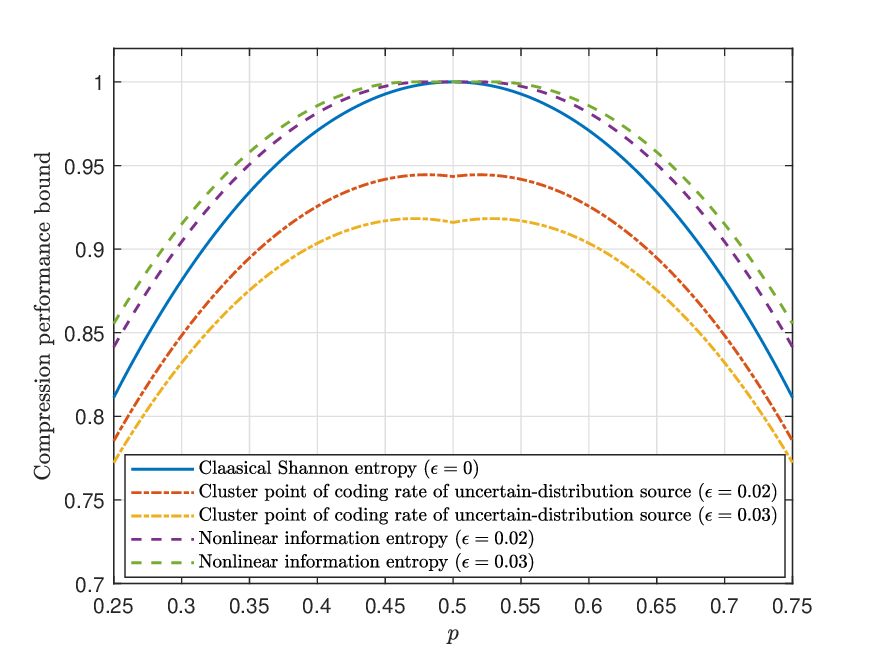}
	\captionsetup{font={scriptsize}}
	\caption{Comparison of the compression performance bound of a uncertain-distribution source and of the corresponding information source without distribution uncertainty.}
	\label{fig_2}
\end{figure}

\section{Conclusion}\label{sec6}

In this paper, we considered the uncertain-distribution sources, whose potential distributions are uncertain. The concepts of nonlinear information entropy are newly defined. Based on the strong laws of large numbers under sublinear expectation, we established the nonlinear source coding theorem, determining the upper bound of the achievable coding rate of uncertain-distribution sources under the maximum error probability criterion, and a cluster point of the coding rate of uncertain-distribution sources under the minimum error probability criterion. Note that when the distributions of source messages become deterministic (i.e., $\Theta$ is a singleton), the conclusions in this study degenerate into the conclusions in classical information theory. This work extends the analysis of information sources from being based on probability theory to being based on nonlinear expectation theory, which are potentially applicable in a wide range of real world problems.

\endgroup

\appendices

\end{document}